\documentclass[a4paper,12pt,reqno,twoside, titlepage]{amsart}
\usepackage[margin=2cm]{geometry}
\usepackage{setspace}
\usepackage{amsmath, mathtools, calligra}
\usepackage{amsfonts}
\usepackage{amssymb}
\usepackage{xcolor}
\usepackage{graphicx}
\usepackage{pstricks}
\usepackage{pstcol}
\usepackage{array}
\usepackage{pictexwd}
\usepackage{graphicx}
\usepackage{subfigure}
\usepackage{fancybox}
\usepackage{newcent}
\usepackage{enumitem}
\usepackage{tikz}
\usepackage{mathrsfs}
\usepackage{bm}
\usepackage{dsfont}

\DeclareMathAlphabet{\mathpzc}{OT1}{pzc}{m}{it}

\setlist{nolistsep}
\newcolumntype{C}[1]{>{\centering\let\newline\\\arraybackslash\hspace{0pt}}m{#1}}
\theoremstyle{plain}
\newtheorem{theorem}{Theorem}

\newtheorem{definition}{Definition}

\newtheorem{lemma}{Lemma}
\newtheorem*{definition*}{Definition}
\newtheorem*{corollary*}{Corollary}
\newtheorem{proposition}{Proposition}

\def \marg {\textrm{marg}}

\def \beq{\begin{eqnarray*}}
\def\eeq{\end{eqnarray*}}

\usepackage{helvet}

\usepackage[textwidth=30mm]{todonotes}

\begin{document}
\title{Robust communication on networks}
\date{\today}
\thanks{The authors gratefully acknowledge the support of the Agence Nationale pour la Recherche under grant ANR CIGNE (ANR-15-CE38-0007-01) and through ORA Project ``Ambiguity in Dynamic Environments'' (ANR-18-ORAR-0005). Laclau gratefully acknowledges the support of the ANR through the program Investissements d'Avenir (ANR-11-IDEX-0003/Labex Ecodec/ANR-11-LABX-0047). }
\author{Marie Laclau}
\address{Marie Laclau, HEC Paris and GREGHEC-CNRS, 1 rue de la Lib\'eration, 78351 Jouy-en-Josas, France}
\email{laclau(at)hec.fr}
\author{Ludovic Renou}
\address{Ludovic Renou, Queen Mary University of London and University of Adelaide, Miles End, E1 4NS, London, UK}
\email{lrenou.econ(at)gmail.com}
\author{Xavier Venel}
\address{Xavier Venel, LUISS Guido Carli University, 32 Viale Romania, 00197 Rome, Italy}
\email{xvenel(at)luiss.it}

\begin{abstract}
We consider sender-receiver games, where the sender and the receiver are two distinct nodes in a communication network. Communication between the sender and the receiver is thus indirect. We ask when it is possible to \emph{robustly} implement the equilibrium outcomes of the  \emph{direct} communication game as equilibrium outcomes of \emph{indirect} communication games on the network. Robust implementation requires that: \emph{(i)} the implementation is independent of the preferences of the intermediaries and \emph{(ii)} the implementation is guaranteed  at all histories consistent with unilateral deviations by the intermediaries. We show that robust implementation of direct communication is possible if and only if either the sender and receiver are directly connected or there exist two disjoint paths between the sender and the receiver. We also show that having two disjoint paths between the sender and the receiver guarantees the robust implementation of all \emph{communication equilibria} of the direct game. We use our results to reflect on organizational arrangements.


\end{abstract}

\maketitle

\newpage

\section{Introduction}

In large organizations, such as public administrations, governments, armed forces and multinational corporations, information typically flows through the different layers of the organization, from engineers, sale representatives, accountants to top managers and executives. Communication is indirect. While indirect communication is necessary in large organizations -- each of the 352, 600 employees of IBM cannot directly communicate with its CEO -- it may harm the effective transmission of valuable information.  Indeed, as the objectives of members of an organization are rarely perfectly aligned, distorting, delaying or even suppressing the transmission of information are natural ways, among others, for members of the organization to achieve their own objectives.\footnote{Companies like Seismic, an  enterprise level software technology provider, specialize in helping teams within organizations to  work better together. See https://hbr.org/2017/02/how-aligned-is-your-organization, ``How Aligned is Your Organization?'' Harvard Business Review, Jonathan Trevor and Barry Varcoe, February 2017. } Does it exist organizational arrangements, which mitigate these issues?  The main insight of this paper is that there is: \emph{matrix organization}.\footnote{We do not claim that matrix organization was designed with this goal in mind; it is rather a consequence, perhaps even unintended, of its design.}\medskip  

Matrix organization (or management) consists in organizing activities along more than one dimension, e.g., function (marketing, accounting, engineering, R\&D, etc), geography (US, Europe, Asia, etc) or products. From the early 60's to the present, large corporations such as NASA, IBM, Pearson, Siemens and Starbucks have adopted this mode of management. A central feature of matrix organization is  \emph{multiple} reporting, that is, low-level employees report to multiple independent managers. E.g., an engineer working on a project to be implemented in Europe will report to the manager of the engineering division, to the manager of the project  and to the manager of the European division. In other words, the information can flow via independent channels. As we shall see, this is the main reason why matrix organization facilitates the effective communication of valuable information. (See Baron and Besanko, 1996, and Harriv and Raviv, 2002, for economic models of matrix organization and Galbraith, 2009,  and Schr\"otter, 2014, for qualitative analysis.) \medskip 

To study the issue of communication in organizations, we consider sender-receiver games on communication networks, where the communication network models the reporting lines within the organization.\footnote{We voluntarily abstract  from a host of other organizational issues, such as team management, remuneration, ownership, or internal capital market.} The main theoretical question this paper addresses is whether it is possible to \emph{robustly} implement the equilibrium outcomes of the  \emph{direct} communication game as equilibrium outcomes of \emph{indirect} communication games on the network.  Robust implementation imposes two requirements: \emph{(i)} the implementation is independent of the preferences of the intermediaries, and \emph{(ii)} the implementation is guaranteed at all histories consistent  with (at most) one intermediary deviating at any stage of the communication game. (Different intermediaries can deviate at different stages.) The first requirement guarantees that the communication is robust to the potential mis-alignment of preferences within the organization. The second requirement  guarantees that the organization can tolerate some mistakes, errors  or even deliberate disruptions in communication. In sum, our two requirements insure effective communication within the organization.\medskip

We assume throughout that the sender and the receiver are not directly connected. (If they are, the problem is trivial.) We prove that the robust implementation of direct communication is possible if, and only if, there are (at least) two disjoint paths of communication between the sender and the receiver. The condition is clearly necessary. Indeed, if it is not satisfied, then all communication between the sender and the receiver must transit through a fixed intermediary (a cut of the graph). This intermediary controls all the flow of information and, therefore, can disrupt the implementation of all equilibrium outcomes of the direct communication game, which are unfavorable to him. To illustrate the difficulties we face, consider the network in Figure \ref{fig:3connexe-intro}.

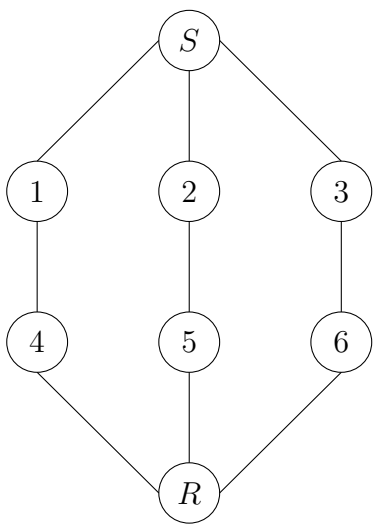
\begin{figure}[h!]
\centering
%
%
%
%
%
%
%
%
%
%
%
%
%
%
%
%
 \begin{tikzpicture}

\draw (4,0) circle [y radius=0.4, x radius=0.4];
\draw (4,2) circle [y radius=0.4, x radius=0.4];
\draw (4,4) circle [y radius=0.4, x radius=0.4];
\draw (4,6) circle [y radius=0.4, x radius=0.4];

\draw (6,2) circle [y radius=0.4, x radius=0.4];
\draw (6,4) circle [y radius=0.4, x radius=0.4];

\draw (2,2) circle [x radius=0.4, y radius=0.4];
\draw (2,4) circle [x radius=0.4, y radius=0.4];

\draw (3.6,0)--(2,1.6);
\draw (2,2.4)--(2,3.6);
\draw(6,3.6)--(6,2.4);
\draw(6,1.6)--(4.4,0); 

\draw(3.6,6)--(2,4.4); 
\draw(4.4,6)--(6,4.4); 

\draw(4,5.6)--(4,4.4); 
\draw(4,3.6)--(4,2.4);
\draw(4,1.6)--(4,0.4);  

\node[] at  (4,0) {$R$};
\node[] at (4,6)  {$S$};
\node[] at  (2,4) {$1$};
\node[] at  (4,4) {$2$};
\node[] at (6,4) {$3$}; 
\node[] at  (2,2) {$4$};
\node[] at  (4,2) {$5$};
\node[] at (6,2) {$6$};

\end{tikzpicture}
\caption{Illustration of the difficulties}
\label{fig:3connexe-intro}
\end{figure}

There are three disjoint paths from the sender to the receiver, so it is tempting to use a majority argument. That is:  to have the sender transmit his message to intermediaries $1$, $2$ and $3$, and to have all intermediaries forward their messages. If the intermediaries are obedient, then the receiver obtains three identical copies of the message sent and indeed learns it. Suppose now that the sender wishes to transmit the message $m$. If  intermediary 1 reports $m' \neq m$ at the first stage and  intermediary $5$ reports $m'' \neq m $ at the second stage, the receiver  then faces the profile of reports $(m',m'',m)$.  Thus, we need the receiver to decode it as $m$. However, the receiver would receive the same profile of reports $(m',m'',m)$ when the sender wishes to transmit the message $m'$, intermediary 3 reports $m$ and intermediary $5$ reports $m''$. Since the receiver would decode it as $m$, he would learn the wrong message. Such a simple strategy does not work in general.\footnote{This simple  strategy works if there are enough disjoint paths. In that example, we need two additional disjoint paths.}  The main difficulty we have to deal with is that different intermediaries deviate at different stages.\medskip

 The communication protocol we construct  requires several rounds of communication and rich communication possibilities. Players must be able to \emph{broadcast} messages to any subset of their neighbors. Broadcasting a message to a group insures common knowledge of the message among the group's members. It is a very natural assumption: face-to-face meetings, online meetings via platforms like Zoom or Microsoft Teams, Whatsapp groups, all makes it possible to broadcast to a group. To sum up, in large organizations, where it is difficult to align the preferences of all, dual lines of reporting facilitate the flow of information throughout the organization and limit the ability of managers to distort reports to their own advantages. However, this comes at a cost: frequent meetings and emailing. For instance, our protocols requires 28 rounds of communication in the illustrative example and more than a hundred meetings! The tendency of matrix organization to generate countless meetings was already noticed in the late 70's. 
\begin{quotation} 
\textit{Top managers were spending more time than ever before in meetings or in airplanes taking them to and from meetings.} (McKinsey Quarterly Review, ``Beyond matrix organization.'' September 1979.\footnote{https://www.mckinsey.com/business-functions/organization/our-insights/beyond-the-matrix-organization\#})  
\end{quotation}

\medskip

We derive another important result. Besides replicating the equilibrium outcomes of the direct communication game, the communication games we construct may have additional equilibrium outcomes. By construction, these equilibrium outcomes are equilibrium outcomes of the mediated  communication game between the sender and receiver, that is, when the communication between the sender and the receiver is intermediated by a mediator. We show that we can replicate \emph{all} equilibrium outcomes of the mediated communication game between the sender and the receiver if, and only if, there are two disjoint paths between the sender and the receiver. Precluding direct communication may therefore help decision makers in achieving better outcomes. 
\medskip

\textit{\textbf{Related literature.}} This paper is related to several strands of literature. The commonality between these strands of literature is the construction of protocols to \emph{securely} transmit a message from a sender to a receiver  on a communication network. The secure transmission of a message requires that (i) the receiver correctly learns the sender's message, and (ii)  intermediaries do not obtain additional information about the  message, while executing the protocol. Reliability (or resiliency) refers to the first requirement, while secrecy refers to the second. 

To start with, there is a large literature in computer science, which studies the problem of secure  transmission of messages on communication networks. See, among others, Beimel and Franklin (1999), Dolev et al. (1993), Linial (1994), Franklin and Wright (2004),  Renault and Tomala (2008), and Renault and al. (2014). This literature provides conditions on the topology of the communication networks for the secure transmission of a message from a sender to a receiver. See Renault and al. (2014) for a summary of these results. An important assumption of all this literature is  that the adversary controls a \emph{fixed} set of nodes throughout the execution of the communication protocols. The adversary we consider is stronger in that it can control different sets of nodes at each round of communication. To the best of our knowledge, such an adversary has not been studied in the computer science literature. However, we restrict attention to singletons, while the computer science literature considers larger sets.

This paper  is also related to the literature on repeated games on networks, where the network models the monitoring structure and/or the communication possibilities. See, among others, Ben-Porath and Kahneman (1996), Laclau (2012, 2014),  Renault and Tomala (1998), Tomala (2011) and Wolitzky (2015). This literature characterizes the networks for which folk theorems exist. An essential step in obtaining a folk theorem  is the construction of protocols, which guarantee that upon observing a deviation,  players start a punishment phase. To do so, when a player observes a deviation, he must be able to securely transmit the message ``my neighbor has deviated'' to all other players. None of these papers have studied the issue of \emph{robustly} communicating the identity of the deviator on \emph{communication networks}. Either the communication is restricted to a network, in which case it is not robust (Laclau, 2012; Renault and Tomala, 1998; Tomala, 2011; Wolitzky, 2015),  or the communication is unrestricted, direct and immediate (Ben-Porath and Kahneman, 1996; Laclau, 2014).\footnote{In the latter case, the network models the structure of observation and/or interaction, but not the communication possibilities. }

 This paper is also related to the literature  on mediated and unmediated communication in games. See, among others, Barany (1985), Ben Porath (1998), Forges (1986, 1990), Forges and Vida (2013), Gerardi (2004), Myerson (1986), Renault and Tomala (2004), Renou and Tomala (2012) and Rivera (2017). In common with this literature, e.g., Barany (1985), Forges (1990), Forges and Vida (2013), or Gerardi (2004), we show that we can emulate mediated communication with unmediated communication. The novelty is that we do it on communication networks, albeit for a particular class of games, i.e., sender-receiver games.

Finally, this paper is related to the large literature on cheap talk games, pioneered by Crawford and Sobel (1984) and Aumann and Hart (1993, 2003). See Forges (2020) for a recent survey. The closest paper to ours is Ambrus et al. (2014). These authors consider simple communication networks (perfectly hierarchical networks) and restrict attention to a particular class of games. Their emphasis is complementary to ours. We ask when it is possible to robustly implement all equilibrium outcomes of the direct (unmediated and mediated) communication game on a communication network, while they ask what are the equilibrium outcomes of their fixed game.

\section{Robust communication}
\subsection{The problem}  We start with some mathematical preliminaries. Unless indicated otherwise, all sets $X$ are  complete separable metric spaces, endowed with their Borel $\sigma$-algebra $\mathbb{B}_X$. We write $\Delta(X)$ for the set of probability measures on $X$.   Let $X$ and $Y$ be two complete separable metric spaces. A probability kernel is a function $f: Y \times \mathbb{B}_X \rightarrow [0,1]$ such that (i) for all $y \in Y$,  $f(y,\cdot): \mathbb{B}_X \rightarrow [0,1]$ is a probability measure, and (ii) for all $B \in \mathbb{B}_X$,  $f(\cdot, B):Y \rightarrow [0,1]$ is measurable. Throughout, we abuse notation and write $f: Y \rightarrow \Delta(X)$ for the  probability kernel $f: Y \times \mathbb{B}_X \rightarrow [0,1]$. \medskip

There are a sender and a receiver, labelled $S$ and $R$, respectively. The sender knows a payoff-relevant state $\omega \in \Omega$, with $\nu\in \Delta(\Omega)$ the prior probability.   The receiver takes an action $a \in A$. For all $i \in \{S,R\}$, player $i$'s payoff function is 
$u_i : A \times \Omega \rightarrow \mathbb{R}$, which we assume to be measurable. 
\medskip 

\textit{\textbf{Direct communication}.} Consider the direct communication game, where the sender can directly communicate with the receiver, that is, the sender can send a message $m \in M$ to the receiver, prior to the receiver choosing an action $a \in A$. In the direct communication game, a strategy for the sender is a map $\sigma: \Omega \rightarrow \Delta(M)$, while a strategy for the receiver is a map $\tau: M \rightarrow \Delta(A)$. We denote $\mathcal{E}^d \subseteq \Delta (\Omega \times A)$ the set of (Nash) equilibrium distributions over states and actions of the direct communication game. Note that we may have $\mathcal{E}^d = \emptyset$.

\medskip

\textit{\textbf{Indirect communication}.} To model indirect communication, we assume that the sender and the receiver are two distinct nodes on an (undirected) network $\mathcal{N}$. The set of nodes, other than $S$ and $R$, is denoted $I$, which we interpret as a set of $n$ intermediaries. Communication between the sender and the receiver transits through these intermediaries.  We let $\mathcal{N}_i$ be the set of neighbors of $i \in I^*:=I \cup \{S,R\}$ in the network. Throughout, we assume that the sender and the receiver are not directly connected in the network $\mathcal{N}$.\medskip

\textbf{\textit{Communication game on a network.}} A communication game on the network $\mathcal{N}$ is a multi-stage game with $T < \infty$ stages, where at each stage players send costless messages to their neighbors and the receiver decides either to take an action (and stop the game) or to continue communicating. 

We first define \emph{communication mechanisms}, denoted $\boldsymbol{\mathcal{M}}$, as the sets of messages players can send to each others.   We allow for a rich set of communication possibilities. Communication can be private, e.g., private emails or one-to-one meetings, or public, e.g., emails sent to distribution lists or group meetings, or a mixed of both. We say that player $i$ {\em broadcasts} a message to a (non-empty) subset of his neighbors $N \subseteq \mathcal{N}_i$ if (i) all players in $N$ receive the same message, and (ii) it is a common belief among all players in $N$ that they have received the same message (in other words, the list of recipients of the message is certifiable among them). Face-to-face group meetings, online meetings via platforms like Zoom or Microsoft Teams, Whatsapp groups, all make it possible to broadcast to a group.

Communication unfolds as follows: at each stage $t$, players \emph{broadcast} messages to all possible (non-empty) subsets of neighbors. The set of messages player $i$ can broadcast to the subset of neighbors $N \in 2^{\mathcal{N}_i}\setminus \{\emptyset\}$ is  $\mathcal{M}_{iN}$. Let $\mathcal{M}_i=\prod_{N \in 2^{\mathcal{N}_i}\setminus \{\emptyset\}}\mathcal{M}_{iN}$ be the set of messages available to player $i$, and $\mathcal{M}=\prod_{i \in  I^*} \mathcal{M}_i$ the set of messages available to all players. Few remarks are worth making. First, private messages correspond to broadcasting to singletons.  Players can thus send private messages to their neighbors. Second, the sets of messages available to a player are independent of both time and past histories of messages sent and received. The latter is with loss of generality. However, without such an assumption, the model has no bite. Indeed, if the only message  player $i$ can transmit upon receiving the message $m$ is the message $m$ itself, an extreme form of history dependence, then we  trivially reproduce  direct communication with indirect communication. 

Finally, we need two additional elements to obtain the communication game from the communication mechanism. First,  we assume that at each stage, the receiver can either take an action $a \in A$ or continue communicating, in which case he sends a message $m_R \in \mathcal{M}_R$. If the receiver takes the action $a$, the game stops. Second, we need to associate payoffs with terminal histories. The payoff to player $i \in I^*$ is $u_i(a,\omega)$ when the state is $\omega$ and the receiver takes action $a$, with $u_i: A \times \Omega \rightarrow \mathbb{R}$ a measurable function. (If the receiver never takes an action, the payoff to all players is $- \infty$.) Thus, communication is purely cheap talk. We denote $\Gamma(\boldsymbol{\mathcal{M}}, \mathcal{N})$ the communication game induced by the mechanism $\boldsymbol{\mathcal{M}}$ on the network $\mathcal{N}$.\bigskip

\textbf{\textit{Strategies and equilibrium.}} A history of messages received and sent by player $i$ up to (but not including) period $t$ is denoted $h_i^t$, with $H_i^t $ the set of all such histories.  A (pure) strategy for player $i \in I $ is a collection of maps $\sigma_i = (\sigma_{i,t})_{t \ge 1}$, where at each stage $t$, $\sigma_{i,t}$ maps $H_i^t$ to $ \mathcal{M}_{i}$. A (pure) strategy for the sender is a collection of maps $\sigma_S = (\sigma_{S,t})_{t \ge 1}$, where at each stage $t$, $\sigma_{S,t}$ maps $\Omega \times H_S^t$ to $ \mathcal{M}_{S}$. A (pure) strategy for the receiver is a collection of maps $\sigma_R= (\sigma_{R,t})_{t \ge 1}$, where at each stage $t$, $\sigma_{R,t}$ maps  $H_R^t$ to $ \mathcal{M}_{R} \cup A$. With a slight abuse of notation, we use the same notation for behavioral strategies.  Let  $H^t=\times_{i \in I^*} H_i^t$ . We write $\mathbb{P}_{\sigma}(\cdot|h^t)$ for the distribution over terminal histories and states induced by the strategy profile $\sigma= (\sigma_S, \sigma_R, (\sigma_i)_{i \in I})$, conditional on  the history $h^t \in H^t$. We write $\mathbb{P}_{\sigma}$ for the distribution, conditional on the initial (empty) history.\medskip

We now define what we mean by ``consistent with unilateral deviations.'' Fix a strategy profile $\sigma$. We define $\Sigma({\sigma})$ as the set of strategy profiles such that $\sigma' \in \Sigma({\sigma})$ if, and only if, there exists a sequence of intermediaries $(i_1,\dots,i_t,\dots)$ such that $\sigma'_t=(\sigma'_{i_t,t},\sigma_{-i_t,t})$ for all $t$. Thus, $\Sigma(\sigma)$ consists of all strategy profiles consistent with at most one intermediary deviating at each stage.  Note that $(\sigma'_i,\sigma_{-i}) \in \Sigma(\sigma)$ for all $\sigma_i'$, for all $i \in I$ since we do not assume that $i_t' \neq i_t$. The same intermediary can deviate multiple times. We let $\mathcal{H}(\sigma)$ be the set of terminal histories consistent with $\Sigma({\sigma})$, that is, $h \in \mathcal{H}(\sigma)$ if there exists $\sigma' \in \Sigma({\sigma})$ such that $h$ is in the support of $\mathbb{P}_{\sigma'}$. Similarly, we let $\Sigma^*({\sigma})$ be the set of strategy profiles consistent with at most one player $i \in I^*$ deviating from $\sigma$ at each stage, thus $\Sigma^*(\sigma)$ also includes deviations by the sender and the receiver. Note that $\sigma \in \Sigma(\sigma) \subset \Sigma^*(\sigma)$.  \medskip

We are now ready to define the concept of \emph{robust communication} on a network. We first start with an informal description. We say that the robust implementation of direct communication is possible on a network if, regardless of the preferences of the intermediaries, we can construct a communication game on the network with the property that for any distribution over actions and states of the direct communication game,  there exists an equilibrium of the communication game which replicates that distribution, not only on the equilibrium path, but also at all paths consistent with at most one intermediary deviating at each stage. In other words, the implementation is not only robust to the preferences of the intermediaries, but also to unilateral deviations.

\begin{definition}[Robust implementation of direct communication on the network $\mathcal{N}$] \label{def:robust-direct}
Robust implementation of direct communication is possible on the network $\mathcal{N}$  if there exists a communication mechanism $\boldsymbol{\mathcal{M}}$ on $\mathcal{N}$ such that for all utility profiles of the intermediaries $(u_i)_{i\in I}$,  for all distributions $\mu \in  \mathcal{E}^d$ of all direct communication games, there exists a Nash equilibrium $\sigma$ of $\Gamma(\boldsymbol{\mathcal{M}}, \mathcal{N})$ satisfying: 
\[\forall \sigma' \in \Sigma(\sigma),\: \: \: \:   \marg_{A \times \Omega} \mathbb{P}_{\sigma'}= \mu.\]
\end{definition}


\subsection{Discussion} 
In modern organizations, most employees, from top-executives to low-level managers, devote a significant fraction of their time to internal communication: they draft and circulate memos, attend and call meetings, write e-mail, etc. The network $\mathcal{N}$ captures these communication possibilities, particularly  who can call a meeting with whom.  If there is a link between players $i$ and $j$ and between players $i$ and $k$, player $i$ can communicate with players $j$ and $k$ both privately (face-to-face meetings) and publicly (group meetings). In organizations, meetings serve several functions, from communicating information to making decisions through generating ideas. The former, i.e., meetings as information forum, is the closest to the role meetings play in our analysis. When player $i$ broadcasts a message to players $j$ and $k$, player $i$ informs players $j$ and $k$.\medskip

We stress that if the robust implementation of direct communication is possible on the network $\mathcal{N}$, then the implementation is not only robust to the preferences of the intermediaries, but also to the underlying communication game. In other words, the very same network of communication $\mathcal{N}$ allows for the robust implementation of \emph{any} equilibrium outcome of \emph{any} direct communication game. This property is important. It implies that, within an organization, the reporting lines do not need to change with the challenges the organization faces. \medskip

Also, the implementation is robust to at most one intermediary deviating at every stage of the indirect communication game. While we cannot hope for the implementation to be robust to any number of deviations at every stage, this is still stronger than imposing that the same intermediary deviates at every stage, as equilibrium analysis would do. \medskip

Finally, the solution concept is Nash equilibrium. As we explain later, it is possible to define belief systems to guarantee the sequential rationality of the equilibria we construct.  All our results remain valid with the concept of (weak) perfect Bayesian equilibrium. We do not explicitly consider refinements of the concept of Nash equilibrium as most refinements do not apply to games with continuous action spaces and, most importantly, this would add nothing to our analysis. None of our analysis requires irrational behaviors off the equilibrium path to sustain the equilibrium. Instead, our analysis relies on indifferences. Robust implementation necessitates that players are made indifferent between deviating or not. (If not, an adversarial intermediary would always benefit from deviating.)

\section{Robust communication}

This section states our results and provide some intuition. All proofs are relegated to the Appendix. We start with our main contribution. 

\subsection{The main result}

\begin{theorem}\label{theo1} Robust implementation of direct communication is possible on the network $\mathcal{N}$ if, and only if, the network $\mathcal{N}$ admits two disjoints paths between the sender and the receiver. 
\end{theorem}

The theorem states that there must exist two disjoint paths between the sender and the receiver. This is clearly necessary. Indeed, if no such paths exist, there exists an intermediary, who controls all the information transiting between the sender and the receiver. In graph-theoretic terms, the intermediary is a \emph{cut} of the graph. In games where the sender and the receiver have perfectly aligned preferences, but the intermediaries have opposite preferences, the ``cut'' can then simulate the histories of messages he would have received in a particular state and behave accordingly. The receiver cannot distinguish between the simulated and real histories and, thus, the ``cut'' can induce the receiver to take a sub-optimal action. In fact, the problem is even more severe. Since robust implementation requires the equilibrium distribution $\mu$ to be implemented at all histories consistent with unilateral deviations by the ``cut,'' the distribution $\mu$ must be independent of the state $\omega$ to guarantee robust implementation. Thus, even in games with perfectly aligned preferences, robust implementation does not hold if there is a ``cut.'' \medskip

To prove sufficiency, we consider the following sub-problem. Suppose that the sender wishes to send the message $m$ to the receiver. We want to find a protocol (a communication mechanism and a strategy profile) such that the receiver correctly learns  the message $m$ not only at all on-path histories, but also at all histories consistent with at most one player, including the sender and the receiver, deviating at each stage of the protocol. Theorem \ref{theo2} states that such a protocol exists when they are two disjoint paths between the sender and the receiver. Moreover, the receiver learns the message after at most $1+(n^C-3)(2n^C-3)$ stages, where $n^C$ is the number nodes on the two shortest disjoint paths from the sender to the receiver. The proof of Theorem \ref{theo1} then follows. Let $(\sigma^*,\tau^*)$ be an equilibrium of the direct communication game, with distribution $\mu$. Let the sender draw the message $m$ with probability $\sigma^*(m|\omega)$ when the state is $\omega$.  We can then invoke Theorem \ref{theo2} to prove that the receiver correctly learns the message $m$ in finite time at all histories consistent with unilateral deviations. Upon learning $m$, we let the receiver choose $a$ with probability $\tau^*(a|m)$. It is then immediate to prove that we have an equilibrium, which guarantees the robust implementation of direct communication.  

\medskip

The solution concept is Nash equilibrium. It is straightforward, albeit tedious and cumbersome, to define belief systems to guarantee the sequential rationality of the equilibria we construct. To see this, let $\sigma$ be a Nash equilibrium of the communication game we construct. Consider any history $h_i^t$ consistent with unilateral deviations, i.e., there exists $\sigma' \in \Sigma(\sigma)$ such that $h_i^t$ is in the support of $\mathbb{P}_{\sigma'}$. If intermediary $i$'s belief at $h_i^t$ is the ``conditioning'' of $\mathbb{P}_{\sigma'}$ on $h_i^t$, then robust implementation implies sequential rationality. Indeed, for all $\sigma' \in \Sigma(\sigma)$, for all $\tilde{\sigma}_i$, the concatenated strategy  $\langle \sigma', (\tilde{\sigma}_i,\sigma_{-i}) \rangle = ((\sigma'_{t'})_{t'<t}, (\tilde{\sigma}_{i,t'},\sigma_{-i,t'})_{t' \geq t})$ is consistent with unilateral deviations, i.e., $\langle \sigma', (\tilde{\sigma}_i,\sigma_{-i}) \rangle \in  \Sigma(\sigma)$. Robust implementation thus implies that $\mathbb{P}_{\langle \sigma', (\tilde{\sigma}_i,\sigma_{-i}) \rangle}=\mu$, that is, $\mathbb{P}_{(\tilde{\sigma}_i,\sigma_{-i})}(\cdot|h^t)\mathbb{P}_{\sigma'}(h^t)=\mu$ for all $\tilde{\sigma}_i$, for all $\sigma' \in \Sigma(\sigma)$. Intermediary $i$ is therefore indifferent between all his strategies at $h_i^t$ (since his belief about $h^t$ is $\mathbb{P}_{\sigma'}(h^t|h_i^t)$).  Since the argument does not rely on the specific $\sigma' \in \Sigma(\sigma)$ we select, we have sequential rationality with respect to any belief system, which is fully supported on the histories consistent with unilateral deviations at $h_i^t$. In other words, as long as the intermediary believes  that at most one player deviated at each of all past stages, we have sequential rationality at $h_i^t$. Similarly, the equilibria we construct are also robust to deviations by the sender and receiver at all stages but the first one, where the sender sends the message, and the last one, where the receiver chooses an action. Therefore, we also have sequential rationality at all histories $h_S^t$ and $h_R^t$.  Finally, at all other histories, we can construct beliefs and actions to guarantee sequential rationality. (See Section \ref{seq-rat} for details.) We now turn to a formal statement of Theorem 2.

\subsection{Strong reliability}
Consider the following problem: the sender wishes to transmit the message $m \in M$, a realization of the random variable $\bold{m}$ with distribution $\boldsymbol{\nu}$, to the receiver, through the network $\mathcal{N}$. We want to construct a \emph{protocol}, i.e., a communication mechanism  and a profile of strategies, such that the receiver correctly ``learns'' the message sent at all terminal histories consistent with unilateral deviations.

\begin{definition}
Transmission of messages is strongly reliable on the network $\mathcal{N}$ if there exist a protocol and a decoding rule $\bold{m}_d: H_{R}^{T+1} \rightarrow M$ such that
\[\mathbb{P}_{\sigma'}\Big(\Big\{h_R^{T+1}: \bold{m}_d(h_R^{T+1})=m\Big\}\Big\vert \bold{m}=m\Big)=1,\] 
for all $\sigma' \in \Sigma^*(\sigma)$, for all $m$. 
\end{definition}

The study of the reliable transmission of messages on networks is not new, see Dolev et al (1993), for an early attempt in Computer Science. (See Renault et al., 2014, for a summary of the literature.) Computer scientists assume that an adversary controls at most $k$ nodes and provide conditions on the network for the reliable transmission of messages. An important feature, however, is that the adversary controls the same $k$ nodes throughout the execution of the protocol. This is a natural assumption in Computer Science, where communication is nearly instantaneous. An adversary would not have the time or capacity to take control of different nodes during the execution of the communication protocol.\footnote{Formally, the reliable transmission of messages requires that $\mathbb{P}_{(\sigma'_i,\sigma_{-i})}\Big(\{h_R^{T+1}: \bold{m}_d(h_R^{T+1})=m\}\Big\vert\bold{m}=m\Big)=1$ for all $\sigma_i'$, for all $i \in I$. This is weaker requirement that strong reliability.} A distinctive feature of our analysis is to consider a \emph{dynamic} adversary, i.e., an adversary which controls a different set of nodes  at each round of the execution of the protocol. (However, we limit our attention to singletons, i.e., $k=1$.) To the best of our knowledge, this is new.

\begin{theorem}\label{theo2}
The transmission of messages is strongly reliable on the network $\mathcal{N}$ if, and only if, the network $\mathcal{N}$ admits two disjoints paths between the sender and the receiver. 
\end{theorem}

We  illustrate our protocol with the help of the network in Figure \ref{fig:losange}. (The protocol we construct is slightly more complicated, but they share the same properties.)

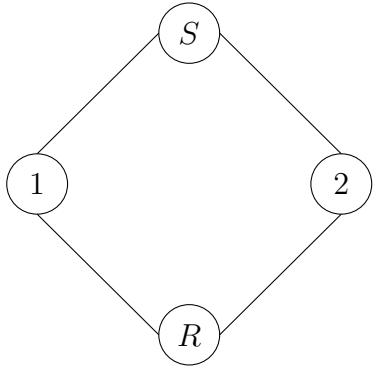
\begin{figure}[h!]
\begin{center}
 \begin{tikzpicture}

\draw (4,0) circle [y radius=0.4, x radius=0.4];
\draw (6,2) circle [y radius=0.4, x radius=0.4];
\draw (4,4) circle [y radius=0.4, x radius=0.4];
\draw (2,2) circle [x radius=0.4, y radius=0.4];

\draw (3.6,0)--(2,1.6);
\draw (2,2.4)--(3.6,4);
\draw(4.4,4)--(6,2.4);
\draw(6,1.6)--(4.4,0); 

\node[] at  (4,0) {$R$};
\node[] at (6,2) {$2$} ;
\node[] at (4,4)  {$S$};
\node[] at  (2,2) {$1$};

\end{tikzpicture}
\end{center}
\caption{Illustration of Theorem \ref{theo2}}
\label{fig:losange}
\end{figure}

The protocol has six stages, which we now describe. To start with, the sender broadcasts the message $m$ to intermediaries $1$ and $2$ at stage $t=1$. At all other stages $t  = 2,\dots,6$, the protocol requires the intermediaries to broadcast the message $m$ and an \emph{authentication key} $x_i^t$, where $x_i^t$ is the authentication key of intermediary $i$ at stage $t$, an uniform draw from $[0,1]$, independent of all messages the intermediary has sent and received. Finally, if the sender observes intermediary $i$ broadcasting a message $m' \neq m$ at stage $t$, the sender broadcasts the triplet $(i,t,x_i^t)$ at stage $t+1$.  We interpret the triplet $(i,t,x_i^t)$ as stating that intermediary $i$ has deviated at stage $t$ and his authentication key is $x_i^t$. If an intermediary receives the triplet $(i,t,x_i^t)$ at stage $t+1$,  the protocol requires the intermediary to broadcast that triplet at all subsequent stages.  \medskip

The receiver does not send messages. At the end of the six stages, the receiver decodes the message as follows. If at any stage, the receiver has received the same message from both intermediaries, then he decodes it as being the correct message. In all other instances, if the receiver has received the \emph{same} message $m_i$ from intermediary $i$ at stages $t_1$, $t_2$ and $t_3$ ($t_1<t_2<t_3$), and  he has \emph{not} received the triplet $(i,t_1,x_{i}^{t_1})$ from the other intermediary by stage $t_3$, then he assumes that the correct message is $m_i$.\footnote{Equivalently, the receiver assumes that the correct message is $m_i$ when he has received $m_i$ from intermediary $i$ at stages $t_1$, $ t_2$ and $t_3$, and all triplets $(i,t_1,y_i^{t_1})$ received from the intermediary $3-i$ by stage $t_3$ are such that $y_i^{t_1}$ is different from $x_i^{t_1}$, the authentication key received from $i$ at stage $t_1$. } \medskip

We now argue that the protocol guarantees the strong reliability of the transmission. To start with, since at most one intermediary deviates at any stage, the protocol guarantees that the receiver obtains at least one sequence of three identical messages from an intermediary. Moreover, if at any stage, the receiver obtains the same message from both intermediaries, it must be the correct message (since at least one intermediary is broadcasting the correct message). So, assume that the receiver obtains the message $m_i$ from intermediary $i$ at stages $t_1$, $t_2$ and $t_3$ $(t_1 <t_2<t_3)$.  If $m_i \neq m$ (hence intermediary $i$ is deviating at stages $t_1$, $t_2$ and $t_3$), the protocol requires the sender to broadcast the triplet $(i,t_1,x_{i}^{t_1})$ at stage $t_1+1 \leq t_2$. The protocol also requires intermediary $j \neq i$ to broadcast the triplet $(i,t_1,x_{i}^{t_1})$ at all stages after having received it.  Hence, the receiver obtains the triplet at stage $t_3$ at the latest. Indeed, since intermediary $i$ is deviating at the stages $(t_1,t_2,t_3)$, intermediary $j \neq i$ cannot be deviating at $t_2$ and $t_3$. Since the authentication key received from intermediary $j$ at either $t_2$ or $t_3$ matches the key received from intermediary $i$ at $t_1$, the receiver learns that the message $m_i$ is not correct. The correct message must therefore be the one broadcasted by intermediary $j$ at stage $t_1$, that is, $m$. Alternatively, if $m_i=m$, the sender does not broadcast the triplet $(i,t_1,x_{i}^{t_1})$. Intermediary $j$ may pretend that the sender had sent the triplet $(i,t_1,y_{i}^{t_1})$ at stage $t_1+1$. However, the probability that the reported authentication key $y_i^t$ matches the actual authentication key $x_i^t$ is zero and, therefore, the receiver correctly infers that the message is $m$. \medskip

We now preview some secondary aspects of the above construction. First, the protocol is robust to deviations by the sender at all stages but the initial stage (where the sender broadcasts $m$).  Indeed, if the sender deviates at stage $t \geq 2$, the two intermediaries don't, and the receiver then correctly learns the message. The protocol we construct shares this property, which will prove useful later on when intermediaries will also have to reliably transmit messages. Second, the protocol starts with the two immediate successors of the sender on the two disjoint paths to the receiver learning the message $m$. At the end of the communication protocol, the receiver also learns the message. In general, the receiver is not the immediate successor of these intermediaries; other intermediaries are. The key step in our construction is to show that at least one of the immediate successors of these intermediaries correctly learns the message at the end of a first block of communication. Therefore, as the protocol goes through blocks, the receiver eventually learns the message. Moreover, each block has $2n^C-3$ stages. (Recall that $n^C$ is the total number of nodes on the two disjoint paths from the sender to the receiver, including them.) Thus, the receiver learns the message in at most $1+(2n^C-3)(n^C-3)$ stages, where the two immediate successors of the sender learns $m$ immediately and then each of the remaining $n^C-3$ players, who do not know $m$ yet, learns the message progressively over time.

\medskip

We conclude this section with some observations about our analysis. First, we allow for rich communication possibilities, most notably, that players are able to broadcast messages to any subset of neighbors. This is necessary for our results to hold. For instance, if  players can only send  private messages (unicast communication), then reliable transmission of messages, let alone strong reliability, is impossible, on the network in Figure \ref{fig:losange}. See Dolev and al. (1993) or Beimel and Franklin (1999). Similarly, if players can only send public messages (broadcast communication), reliable transmission of messages, let alone strong reliability, is impossible on the network in Figure \ref{fig:losange-complete}. See Franklin and Wright (2000) and Renault and Tomala (2008). 

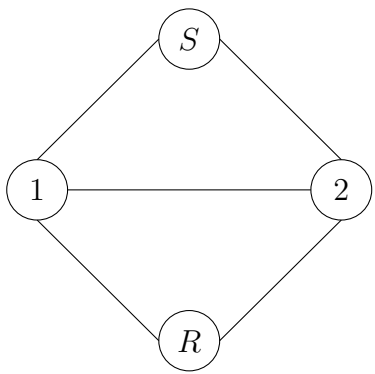
\begin{figure}[h!]
\begin{center}
 \begin{tikzpicture}

\draw (4,0) circle [y radius=0.4, x radius=0.4];
\draw (6,2) circle [y radius=0.4, x radius=0.4];
\draw (4,4) circle [y radius=0.4, x radius=0.4];
\draw (2,2) circle [x radius=0.4, y radius=0.4];

\draw (3.6,0)--(2,1.6);
\draw (2,2.4)--(3.6,4);
\draw(4.4,4)--(6,2.4);
\draw(6,1.6)--(4.4,0); 
\draw(2.4,2)--(5.6,2);

\node[] at  (4,0) {$R$};
\node[] at (6,2) {$2$} ;
\node[] at (4,4)  {$S$};
\node[] at  (2,2) {$1$};

\end{tikzpicture}
\end{center}
\caption{Broadcasting to all: reliable communication is impossible}
\label{fig:losange-complete}
\end{figure}

While the formal proofs of these two impossibilities are involved, the intuition is that the receiver is unable to distinguish between two types of histories: histories where intermediary 1 is pretending that the message is $m'$ and intermediary 2 is deviating, and histories where intermediary 2 is pretending that the message is $m$ and intermediary 1 is deviating.  The key is that an intermediary can simulate fictitious histories, i.e., histories of messages sent and received when the message is any $m$, and behave accordingly. As is clear from the example, the protocol we construct makes it possible for the receiver to distinguish these two types of histories. If intermediary 1 deviates and pretend that the message is $m'$, the receiver correctly infers that intermediary 2 is not deviating. This requires to broadcast messages to selected subsets of neighbors.

Second, our protocol does not restrict the messages players can broadcast to any subset of their neighbors. E.g., in addition to the messages our protocol requires the intermediaries to broadcast, the intermediaries can also send private messages to the sender and the receiver. The equilibrium we construct simply treats these additional messages as uninformative (babbles). 

Third, we use authentication keys. While their use is ubiquitous in online retailing, it is less so in the daily activities of most organizations. Their only purpose, however, is to insure that a group of individuals share a common information, which can only be known by individuals outside the group if it is told to them by one member of the group. For instance, at a meeting, the common information can be the color of the tie of the meeting's chair or the identity of the second speaker. 

Fourth, a detailed knowledge of the communication network is not needed. To execute our protocol, a player on one of the path from the sender to the receiver only needs to know his two immediate neighbors on the path and the total number of players on the two disjoint paths.

Finally, the communication games we construct may have additional equilibrium distributions. By construction, these distributions correspond to communication equilibria of the direct communication game. The next section shows that in fact we can obtain \emph{all} the communication equilibrium distributions of the direct communication game.

\subsection{Mediated communication }\label{comeq}
Consider the following mediated extension of the direct communication game. The sender first sends a message $m \in M$ to a mediator. The mediator then sends a message $r \in R$, possibly randomly, to the receiver, who then takes an action $a \in A$. A strategy for the sender is a map $\sigma:\Omega \rightarrow \Delta(M)$, while a strategy for the receiver is a map $\tau:R \rightarrow \Delta(A)$. The mediator follows a recommendation rule: $\varphi: M \rightarrow \Delta(R)$. A communication equilibrium is a communication device $\langle M,R, \varphi \rangle$ and an equilibrium $(\sigma^*,\tau^*)$ of the mediated game induced by the communication device. Thanks to the revelation principle (Forges, 1986 and Myerson, 1986), we can restrict attention to canonical communication equilibria, where $M=\Omega$, $R=A$ and the sender has an incentive to be truthful (to report the true state), and  the receiver has an incentive to be obedient (to follow the recommendation). We let $\mathcal{CE}^d$ be the set of communication equilibrium distributions over $A \times \Omega$ of the direct communication game. \medskip

For an example, consider the game in Figure \ref{tab:les-pigeons} due to Farrell (1988), where both states are equally likely.  
\begin{figure}[h!]
\begin{center}
\begin{tabular}{|c|c|c|c|}
\hline
 $(u_S,u_R)$ & $a$ & $b$ & $c$ \\ \hline
$\omega$ & $2,3$& $0,2$ & $-1,0$ \\ \hline
$\omega'$ & $1,0$ & $2,2$ & $0,3$\\ \hline 
\end{tabular}
\end{center}
\caption{An example (Farrell, 1988)}
\label{tab:les-pigeons}
\end{figure}
\medskip 

Farrell proves that the receiver takes action $b$ in all equilibria of the direct communication game. The sender's payoff is therefore $0$ (resp., $2$) when the state is $\omega$ (resp., $\omega'$), while the receiver's payoff is $2$.   As argued by Myerson (1991), there exists a communication equilibrium, where both the receiver and the sender are better off. To see this, assume that the mediator recommends action $b$ to the receiver at state $\omega'$ and randomizes uniformly between the recommendations $a$ and $b$ at $\omega$. Upon observing the recommendation $a$, (resp., $b$), the receiver infers that the state is $\omega$ with probability $1$ (resp., $1/3$) and, therefore, has an incentive to be obedient. Similarly, the sender has an incentive to truthfully report his private information. The sender's payoff is therefore $1$ (resp., $2$) when the state is $\omega$ (resp., $\omega'$), while the receiver's payoff is $9/4$. Both the sender and receiver are better off. We now argue that if the communication between the sender and the receiver is intermediated (not to be confused with mediated), then it is possible to replicate this equilibrium outcome, without the need for a trusted mediator, and regardless of the preferences and behaviors of the intermediaries (provided that at most one intermediary deviates at each stage of the communication protocol).

\medskip

\begin{theorem}\label{theo3} Robust implementation of mediated communication is possible on the network $\mathcal{N}$ if, and only if,  
the network $\mathcal{N}$ admits two disjoints paths between the sender and the receiver. 
\end{theorem}

Theorem \ref{theo3} thus states that if there are two disjoint paths between the sender and the receiver in the network, then robust implementation of mediated communication is possible. In other words, we can replicate the mediator through unmediated communication, and the replication is robust.\medskip

To get a flavor of our construction, let us again consider the network in Figure \ref{fig:losange}. Assume that $\Omega$ and $A$ are finite sets. Fix a canonical communication equilibrium $\varphi: \Omega \rightarrow \Delta(A)$. For all $\omega \in \Omega$, let $A_{\omega}$ be a partition of $[0,1]$ into $|A|$ subsets, with the subset $A_{\omega}(a)$ corresponding to $a$ having Lebesgue measure $\varphi(a|\omega)$.  The protocol has three distinct phases. In the first phase, the sender broadcasts the state $\omega$ to the intermediaries $1$ and $2$. The second phase replicates the communication device. To do so, the sender and intermediary 1  simultaneously choose a randomly generated number in $[0,1]$. Let $x$ and $y$ be the numbers generated by the sender and intermediary 1, respectively. Players then follow the communication protocol constructed above (see Theorem \ref{theo2}), which makes it possible for intermediary 1 to reliably transmit $y$ to intermediary 2.\footnote{Therefore, intermediary 1 plays the role of the sender, while intermediary 2 plays the role of the receiver.} Thus, at the end of the second phase, the sender and both intermediaries know $\omega$, $x$ and $y$, while the receiver only knows $y$.  The third phase starts once the sender and both intermediaries have learnt $x$ and $y$. At the first stage of the third phase, the sender and the intermediaries simultaneously  compute $x+y \pmod{[0,1]}$, output the recommendation $a$ if $x+y \pmod{[0,1]} \in A_{\omega}(a)$, and each starts a  communication protocol  to reliably transmit the  recommendation to the receiver. Thus, the three communication protocols are synchronized and start at the very same stage. At the end of the third phase, the receiver learns the recommendation sent by the sender and both intermediaries. Since at most one of them can deviate at the stage where they broadcast their recommendation, the receiver decodes at least two identical recommendations and plays it. 

\section{Concluding remarks}

Our analysis extends to communication games with multiple senders. More precisely, consider a direct communication game, where senders receive  private signals about a payoff-relevant state, send messages to the receiver, and the receiver takes an action.  If there exist two disjoint paths of communication from each sender to the receiver, we can then replicate our analysis to robustly implement any equilibrium distribution of the direct communication game. The key is to have all senders broadcast their messages at the first stage and then to run copies of our protocol in parallel. Since our protocol is resilient to what a player learns during its execution, this guarantees that the receiver learns the correct messages. We stress, however, that it is essential that all senders move simultaneously at the first stage. \medskip 

Throughout, we have assumed that the sender and the receiver are not directly connected. We now discuss how our results would change if the sender and the receiver can also communicate directly. Theorem \ref{theo1} would extend immediately. More precisely, we would have that robust implementation of direct communication is possible if, and only if, either the sender and receiver are directly connected or there exist two disjoint paths between the sender and the receiver.

The extension of Theorem \ref{theo3} is more delicate. To start with,  notice that if the direct link is the only link between the sender and the receiver, then there is no hope to replicate the mediator, as all communication would be direct. We need to be able to intermediate the communication.  We claim that if there are at least two intermediaries, labelled $1$ and $2$, such that the sender, the receiver and the two intermediaries are on a ``circle,'' then robust implementation of mediated communication is possible.\footnote{A circle is a collection of nodes such that all pairs of nodes have two disjoint paths between them.}

We now present an informal proof. As a preliminary observation, note that  on the circle, it must be that either the two intermediaries are on two disjoint paths from the sender to the receiver or are on the same path. See the two networks in Figure \ref{fig:direct-link} for an illustration.

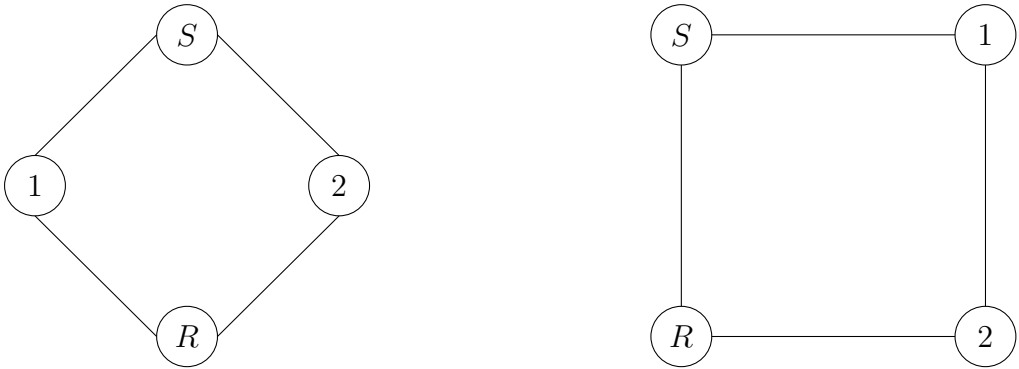
\begin{figure}[h]
\centering
\begin{minipage}{.5\textwidth}
\centering
 \begin{tikzpicture}

\draw (4,0) circle [y radius=0.4, x radius=0.4];
\draw (6,2) circle [y radius=0.4, x radius=0.4];
\draw (4,4) circle [y radius=0.4, x radius=0.4];
\draw (2,2) circle [x radius=0.4, y radius=0.4];

\draw (3.6,0)--(2,1.6);
\draw (2,2.4)--(3.6,4);
\draw(4.4,4)--(6,2.4);
\draw(6,1.6)--(4.4,0); 

\node[] at  (4,0) {$R$};
\node[] at (6,2) {$2$} ;
\node[] at (4,4)  {$S$};
\node[] at  (2,2) {$1$};

\end{tikzpicture}
\end{minipage}%
\begin{minipage}{.5\textwidth}
\centering
 \begin{tikzpicture}

\draw (6,4) circle [y radius=0.4, x radius=0.4];
\draw (6,0) circle [y radius=0.4, x radius=0.4];
\draw (2,4) circle [y radius=0.4, x radius=0.4];
\draw (2,0) circle [x radius=0.4, y radius=0.4];

\draw (6,3.6)--(6,.4);
\draw (2,3.6)--(2,.4);
\draw(2.4,4)--(5.6,4);
\draw(2.4,0)--(5.6,0); 

\node[] at  (2,0) {$R$};
\node[] at (6,0) {$2$} ;
\node[] at (2,4)  {$S$};
\node[] at  (6,4) {$1$};

\end{tikzpicture}
\end{minipage}
\caption{$\mathcal{N}$ (left) and $\mathcal{N}^*$ (right)}
\label{fig:direct-link}
\end{figure}

In the former case (network $\mathcal{N}$), Theorem \ref{theo3} applies verbatim. In the latter case (network $\mathcal{N}^*$), we need to modify the protocol of Theorem \ref{theo3} to guarantee that intermediary 2 learns the state $\omega$, without the receiver learning it. Once the sender and the two intermediaries know $\omega$, we can then use the second and third phase  of the protocol constructed in Section 3.3 to implement the communication equilibrium. 

We modify the first phase as follows. We first let intermediary $1$ broadcast an encryption key $k$  to the sender and intermediary 2. The sender then encrypts the state $\omega$ with the encryption key $k$ and transmits the encrypted message to intermediary 2. The transmission of the encrypted message is achieved via the strongly reliable protocol constructed in Section 3.2. 

It remains to argue that at the end of the first phase, the sender and the two intermediaries know $\omega$. It is clear for the sender and intermediary 1. As for intermediary 2, at the end of the first phase, he knows the encryption key and the encrypted message and, thus learn the state $\omega$. This completes the informal ``proof.'' 

Note that either the sender or an intermediary may attempt to reveal the state to the receiver. We do not preclude this possibility. However, as already explained, the equilibria we construct are such that the receiver does not expect the sender and intermediaries to do so and thus consider their attempts as gibberish. 

Finally, we don't know whether the condition of having at least two intermediaries such that the sender, the receiver and the two intermediaries are on a ``circle,'' is necessary. We conjecture that it is.

\appendix 

\section*{Appendix}\label{2A1}
We first prove Theorem \ref{theo2} and then prove Theorems \ref{theo1} and \ref{theo3}. We do not prove the necessity parts as the proofs follow well-trodden paths, see e.g.,  Renault and Tomala (2008) or Renault et al. (2014).

\section{Proof of Theorem \ref{theo2}}

Assume that  the network $\mathcal{N}$ admits two disjoint paths between the sender and the receiver, and denote the two shortest paths by $S,i_1,\dots,i_{K},R$ and $S,j_1,\dots,j_{K'},R$ respectively, for some $K, K' \ge 1$. We let $\mathcal{P}$ be the set of nodes on these two paths, including the sender and the receiver, and let $n^C$ be its cardinality. Throughout, we refer to these two paths as the ``circle,'' with the nodes $\{S,i_1,\cdots,i_K,R\}$ (resp.,$\{S,j_1,\dots,j_{K'},R\}$) as the ``left side of the circle'' (resp., ``right side of the circle'').  For each player $p \in \mathcal{P} \setminus \{S,R\}$, we call the {\em successor} of $p$, denoted $p^+$, his immediate successor on the path to the receiver.  Similarly, we call the {\em predecessor} of $p$, denoted $p^-$ his immediate predecessor on the path to the sender. E.g., if $p=i_{k}$ for some $1 \leq k\leq K$, $p^+=i_{k+1}$ and $p^-=i_{k-1}$, with the convention that $i_0=S$ and $i_{K+1}=R$.  The sender (resp., the receiver) have the same two nodes as predecessors and successors, i.e., $i_1$ and $j_1$ (reps., $i_K$ and $j_{K'}$). However, as will be clear later, whenever $i_1$ (resp., $i_K$) plays the role of a predecessor, then $j_1$ (resp., $j_{K'}$) plays the role of the successor and, conversely. Thus, we mostly focus on messages flowing from the sender to the receiver. However, messages will also need to flow in the other direction.


\subsection{The communication protocol $(\boldsymbol{\mathcal{M}},\sigma)$}

Throughout, we write $[1:T]$ for $\{1,\dots,T\}$.

 \textit{\textbf{The message space.}} Remember that $M$ is the set of messages in the direct communication game and let $m_0 \notin M$ be an arbitrary message, interpreted as a null message. The set of messages player $i \in I^*$ can broadcast to the subset $N_i \in 2^{\mathcal{N}_i}\setminus \{\emptyset\}$ is:
  \begin{eqnarray*}
 \mathcal{M}_{i,N_i} & = & \Big(M \cup \{m_0\}\Big)  \times [0,1]  \\
 & &\bigtimes   \Big(\bigtimes_{j\in \mathcal{N}_i} \Big \{ \{\emptyset\} \bigcup \big(\{j\}\times [1:T] \times [0,1]\big)\Big\} \Big)\\
 & &\bigtimes   \Big(\bigtimes_{j \in I^* \setminus \mathcal{N}_i} \Big \{ \{\emptyset\} \bigcup \big(\{j\}\times [1:T] \times [0,1]\big)\Big\}^{L} \Big),
 \end{eqnarray*}
with $L=1+(n^C-3)(2n^C-3)$.

The set of messages player $i$ can send is $\bigtimes_{N_i \in  2^{\mathcal{N}_i}\setminus \{\emptyset\}}\mathcal{M}_{i,N_i}$.  In words,  each player can broadcast to any subset of his neighbors a grand message composed of: (i) a message $m \in M$ or the null message $m_0$, (ii) a number in $[0,1]$ and (iii) a tuple of triplets, each of them being composed of the name of a player, a stage, and a number in $[0,1]$. Crucially, player $i$ can send a \emph{single} triplet about each of his neighbors at each stage. However, player $i$ can send several triplets (at most $L$) about the other players. In other words, at each stage, player $i$ can send a list of triplets to his neighbors $N_i$, but no list includes more than one triplet about $j \in \mathcal{N}_i$. 

\medskip

%
%
%
%
%

\textit{\textbf{The strategies of the players.}} For any player $p \notin \mathcal{P}$, the protocol requires them to broadcast randomly drawn messages in $\mathcal{M}_{i,N_i}$ at each stage $t$ to each subset of neighbors $N_i \in 2^{\mathcal{N}_i} \setminus \{\emptyset\}$, independently of all messages received and sent up to stage $t$. In words, they are babbling.

We now define the strategy for player $p \in \mathcal{P}$. We focus on the messages they broadcast to their neighbors on the circle, i.e., to $\mathcal{N}_p \cap \mathcal{P}$. To all other subsets of neighbors, they sent randomly generated messages, independently of the histories of messages sent and received, i.e., they babble. In what follows, when we say that player $p$ broadcasts a message, we mean that he broadcasts the message to the subset $\{p^-,p^+\}$. Remember that for $S$ (resp., $R$), $\{p^-,p^+\} = \{i_1,j_1\}$ (resp., $\{i_{K},j_{K'}\}$).

\begin{itemize}
\item {\bf Authentication keys:} At each stage $t \ge 1$ of the communication protocol, $p$ broadcasts a uniformly drawn message $x^t_{p}$ in $[0,1]$ to his neighbors on the ``circle,'' that is, to  the two players in $\mathcal{N}_p \cap \mathcal{P}$: this message $x^t_{p}$ is called the {\em authentication key} of player $p$ at stage $t$.
\item {\bf First stage:} At stage $t=1$, the sender broadcasts the message $m$ to his neighbors on the circle, i.e., to $i_1$ and $j_1$, along with his authentication key. All other players broadcast $m_0$ to their neighbors on the circle (along with their authentication keys).
\item {\bf Subsequent stages:} Starting from stage $t=2$ onwards, the protocol proceeds in blocks of $2n^C-3$ stages. Denote these blocks by $B_b$, with $b=1,2\dots, \overline{B}$. (We have at most $n^C-3$ blocks.) For instance, $B_1$ stands for the block that starts at stage $t=2$, $B_2$ stands for the block that starts at stage $t=2+2n^C-3=2n^C-1$, etc. For each $b=1,2\dots$, let $B_b:=\{t_b,\dots,t_b+2n^C-3\}$ where $t_b$ is the first stage of block $B_b$. In each block $B_b$, the strategy of $p$ is the following:
\begin{itemize}
\item {\bf Transmission of the sender's message:} 
\begin{itemize}
\item if $p$ knows  the message $m$ at the beginning of the block, that is, either $p$ is in $\{S,i_1,j_1\}$ or $p$ has {\em learnt} the message $m$ at the end of the previous block (see below, where the decoding rule at the end of each block is defined), then $p$ broadcasts the message $m$ to his neighbors $p^-$ and $p^+$ at all stages $t_b,\dots,t_b+{2n^C-3}$ of the current block (the neighbors of $S$ are $i_1$ and $j_1$);
\item if $p$ does not know the message $m$ at the beginning of the block, then $p$ broadcasts $m_0$ to his neighbors $p^-$ and $p^+$ at all stages $t_b,\dots,t_b+{2n^C-3}$ of the current block.
\end{itemize}
(Remember that $p$ also sends an authentication key.)
\item {\bf Detection of deviations:} 
\begin{itemize}
\item if $p$ detects his successor $p^+$ making a false announcement about the message $m \in M$ at some stage $t \in \{t_b,\dots,t_b+{2n^C-4}\}$, that is,
\begin{itemize}
\item either $p$ knows $m \in M $ and $p^+$ broadcasts at stage $t$ the message $\tilde{m} \in M \setminus \{m\}$, interpreted as ``{\em player $p^+$ is broadcasting the false message $\tilde{m}$},'' 
\item or $p$ does not know the message $m$  and $p^+$ broadcasts at stage $t$ the message $\tilde{m} \in M $, interpreted as  {\em ``player $p^+$ is broadcasting the message $\tilde{m}$ although he cannot know it,''}
\end{itemize}
\noindent  then $p$ broadcasts the triplet $(p^+,t,x_{p^+}^t)$ to player $p^-$ and to $p^+$, where $x^t_{p^+}$ is the true authentication key broadcasted by $p^+$ at $t$.
Note that if $p=S$ (resp., $R$), then $p^+$ is either $i_1$ or $j_1$ (resp., $i_K$ or $j_{K'}$). 
\item if $p$ does not detect his successor $p^+$ making a false announcement about the message $m \in M$ at some stage $t \in \{t_b,\dots,t_b+{2n^C-4}\}$, then $p$ broadcasts the triplet $(p^+,t,y)$ to $p^-$ and $p^+$, where $y$ is randomly drawn from $[0,1]$.

\end{itemize}
The key observation to make is that only players $p$ and $p^{++}$ know the true authentication key of $p^+$ at stage $t$. Therefore, $p^{++}$ can authenticate whether $p^+$ deviated at some stage $t$ by cross-checking the authentication key $x^t_{p^+}$ received from $p^+$ at $t$ with the key broadcasted by $p$ (and having transited on the circle in the opposite direction).

\item {\bf Transmission of past deviations:}
\begin{itemize}
\item If $p \neq i_1$ is on the left side of the circle and receives at some stage $t\in \{t_b,\dots,t_b+{2n^C-4}\}$
\begin{itemize}
\item  from $p^+$  a message containing the triplet $(p',s,x_{p'}^s)$, $t_b\le s < t$ with $p'$ on the left side of the circle, then $p$ broadcasts the message to $p^{-}$ and $p^+$ at stage $t+1$.
\item  from $p^-$  a message containing the triplet $(p',s,x_{p'}^s)$, $t_b\le s < t$ with $p'$ on the right side of the circle, then $p$ broadcasts the message to $p^-$ and $p^{+}$ at stage $t+1$.
\end{itemize}
\item Similarly, if  $p \neq j_1$ is on the right side of the circle and receives at some stage $t\in \{t_b,\dots,t_b+{2n^C-4}\}$
\begin{itemize}
\item  from $p^+$  a message containing the triplet $(p',s,x_{p'}^s)$, $t_b\le s < t$ with $p'$ on the right side of the circle, then $p$ broadcasts the message to $p^{-}$ and $p^+$ at stage $t+1$.
\item  from $p^-$  a message containing the triplet $(p',s,x_{p'}^s)$, $t_b\le s < t$ with $p'$ on the left side of the circle, then $p$ broadcasts the message to $p^{-}$ and  $p^{+}$ at stage $t+1$.
\end{itemize}
\item If $p=i_1$ (respectively $p=j_1$) receives from $p^+$ a message containing the triplet $(p',s,x_{p'}^s)$, $t_b\le s < t$ with $p'$ on the left side of the circle (respectively with $p'$ on the right side of the circle), then $p$ broadcasts the message to $p^{-}$ and $p^+$ at stage $t+1$.
\item If $p=i_1$ (respectively $p=j_1$) receives from $p^- = S$ at some stage $t\in \{t_b,\dots,t_b+{2n^C-4}\}$ a message containing the triplet $(p',s,x_{p'}^s)$, with $t_b\le s < t $ and $p'$ on the right side of the circle (resp., on the left side of the circle), then two cases are possible:

\begin{itemize}
\item[(i):] If $p' \neq j_1$ (resp., $p' \neq i_1$), then $p$ broadcasts it to both $p^-=S$ and $p^+=i_2$ (resp., $p^+=j_2$) at stage $t+1$.
\item[(ii):]  If $p' = j_1$ (resp., $p'=i_1$) and $p=i_1$ (resp., $p=j_1)$ has not received the triplet $(p,s,x_p^s)$, then $p=i_1$ (resp., $p=j_1$) broadcasts the triplet $(j_1,s,x_{j_1}^s)$ (resp., $(i_1,s,x_{i_1}^s)$) at stage $t+1$.
\item[(iii):] If $p' = j_1$ (resp., $p'=i_1$) and $p=i_1$ (resp., $p=j_1)$ has received the triplet $(p,s,x_p^s)$, then $p=i_1$ (resp., $p=j_1$) broadcasts the triplet  $(p',s,y)$ at state $t+1$, with $y$ a random draw from $[0,1]$. 
\end{itemize}
\smallskip

The intuition is that if $p$ receives from $S$ a message, which reads as {\em ``$S$ claims that both $i_1$ and $j_1$ deviated at the same stage $s$,''} then $S$ must be deviating (since under unilateral deviations, at most one player deviates at each stage).

\item {\bf Auto-correcting past own deviations:} if $p$ has received the triplet $(p',s,x_{p'}^s)$ at stage $t \in \{t_b,\dots,t_b+{2n^C-4}\}$ but didn't forward it at stages $t+1,\dots, t+\Delta$, then he forwards it at stage $t+\Delta+1$. 
In words, the protocol requires a player to broadcast the triplet $(p',s,x_{p'}^s)$ at stage $t+1$ upon receiving it at stage $t$, to broadcast it at $t+2$ if he fails to broadcast it at stage $t+1$, to broadcast it at $t+3$ f he fails to broadcast it at stage $t+1$ and $t+2$, etc, so that unless the player deviates at all stages $t'\geq t+1$, the triplet is broadcasted at some stage during the block. 
\end{itemize}
\end{itemize}
\end{itemize}

\medskip 

\textbf{\textit{The decoding rule.}} The decoding rule describes how messages are analyzed at the end of each block. Players not in $\mathcal{P}$ do not analyze their messages. Consider now the players in $\mathcal{P}$.

At the beginning of block $B_1$, the sender and his two neighbors $i_1$ and $j_1$ {\em know} the message $m$ broadcasted by the sender at stage $t=1$. At the end of each block, only players who do not know yet $m$ analyze the message received during the block. Thus, only the players in $\mathcal{P}\setminus \{S,i_1,j_1\}$ analyze messages as the end of the block $B_1$. (The gist of our arguments is to show that the set of players who know $m$ at the end of a block is strictly expanding over time and, ultimately, includes the receiver.) Thus, consider player $p$, who does not know yet $m$ at the beginning of the block $B_b$. At the end of the block $B_b$, he analyzes messages as follows: 

\begin{itemize}
\item If $p$ has received during the block $(n^C-1)$ times a grand message containing the same message $m \in M$ from his predecessor $p^-$, let say at stages $s^1,\dots,s^{n^C-1}$, where $t_b \le s^1 < s^2 < \dots < s^{n^C-1} \le t_b+{2n^C-3}$,
\item and if $p$ has not received by stage $s^{n^C-1}$ at the latest from his successor $p^+$ the message $(p^-, s^1, x_{p^-}^{s^1})$ where $x_{p^-}^{s^1}$ matches the value of the authentication key received by $p$ from $p^-$ at stage $s_1$,
\end{itemize}

\noindent then, player $p$ {\em learns} the message $m$ and starts the next block $B_{b+1}$ as a player who {\em knows} $m$. Otherwise, player $p$ does not learn the message. Moreover, once a player knows the message $m$, he knows it at all the subsequent blocks. If $p=R$ and $p^-=i_K$ (resp, $j_{K'}$), then $p^+=j_{K'}$ (resp., $p^{-}=i_{K}$).

At all other histories, the strategies are left unspecified.

\subsection{Two key properties of the protocol}

The protocol we construct has two key properties. The first property states that no player $p \in \mathcal{P}$ learns incorrectly, that is, if player $p \in \mathcal{P}$ learns a message, the message is indeed the one the sender has sent. Lemma \ref{lem1} is a formal statement of that property.  

\begin{lemma}\label{lem1}Let $m \in M$ be the message broadcasted by the sender to $i_1$ and $j_1$ at stage $t=1$. If at most one player deviates from the protocol at each stage, then it is not possible for player $p \in \mathcal{P}$ to learn $m' \in \mathcal{M}\setminus \{m\}$. 
\end{lemma}

\begin{proof}[ Proof of  Lemma \ref{lem1}] By contradiction, assume that player $p \in \mathcal{P}$ learns $m' \in M \setminus \{m\}$ at the end of some block $B_b$, $b \ge 1$.  Without loss of generality, assume that $p$ is the first player to learn $m'$ on the path from $S$ to $R$ where $p$ lies. 

For player $p$ to learn $m'$, during the block $B_b$, it must be that player $p$ has received a sequence of $n^C-1$ grand messages from his predecessor $p^-$, say at stages $s^1,\dots,s^{n^C-1}$ with $t_b \le s^1 < s^2 < \dots < s^{n^C-1} \le t_b+{2n^C-1}$, such that (i) the first element of each of the $n^C-1$ grand messages is $m'$ and (ii) player $p$ did not receive from $p^+$ the triplet $(p^-, s^1, x_{p^-}^{s^1})$ on or before stage $s^{n^C-1}$, where $x_{p^-}^{s^1}$ matches the value of the authentication key received  from $p^-$ at stage $s^1$.

Since we assume that $m' \ne m$, it must be that that $p^-$ is deviating at all stages $s^1, \dots,s^{n^C-1}$, as we assume that $p$ is the first player to learn $m'$. Therefore, since we consider at most one deviation at each stage, all players in $\mathcal{P} \setminus \{p^-\}$ are playing according to $\sigma$ at all stages $s^1, \dots, s^{n^C-1}$. It follows that  player $p^{--}$, the predecessor of $p^-$, broadcasts the triplet $(p^-, s^1, x_{p^-}^{s^1})$ to $p^-$ and $p^{---}$ at stage $s^2$ at the latest, that player $p^{---}$ broadcasts it at stage $s^3$ at the latest, etc.\footnote{Notice that it is possible for $p^{--}$ to broadcast the message $(p^-, s^1, x_{p^-}^{s^1})$ before stage $s^2$. For instance it is possible to have a stage $s^{1'}$, with $s^1 < s^{1'} < s^2$, such that (i) $p^-$ is not deviating at stage $s^{1'}$ and sends either $m$ or $m_0$ depending on if he knows $m$ or not, and (ii) $p^{--}$ is deviating at stage $s^{1'}$ by not transmitting $(p^-, s^1, x_{p^-}^{s^1})$ to $p^-$ and $p^{--}$.}

 Since there are $n^C-2$ nodes other than $p^-$ and $p$ on the circle, player $p^+$ broadcasts the triplet $(p^-, s^1, x_{p^-}^{s^1})$ to players $p$ and $p^{++}$ at stage $s^{n^C-1}$ at the latest. Thus, player $p$ does not validate $m'$ with that sequence of grand messages.
 
 Since this is true for any such sequences, player $p$ does not learn $m'$ at block $B_b$.
\end{proof}

Lemma \ref{lem1} states that no player on the circle learns an incorrect message. The next Lemma states that at least one new player learns the correct message at the end of each block, which guarantees that the receiver learns the correct message at the latest after $1+(n^C-3)(2n^C-3)$ stages.

\begin{lemma}\label{lem2}Let $m \in M$ be the message broadcasted by the sender to $i_1$ and $j_1$ at stage $t=1$.
Suppose that all intermediaries $i_1$ to $i_k$  and $j_1$ to $j_{k'}$ know the message $m$ at the beginning of the block $B_b$.  If at most one player deviates from the protocol at each stage, then either intermediary $i_{k+1}$ or intermediary $j_{k'+1}$ is learning $m$ at the end of the block $B_b$ with probability one.\end{lemma}

To be more precise, Lemma \ref{lem2} states that for all $\sigma' \in \Sigma^*(\sigma)$, the subset of histories at which either intermediary $i_{k+1}$ or intermediary $j_{k'+1}$ is learning $m$ at the end of the block $B_b$ has probability one according to $\mathbb{P}_{\sigma'}$.

\begin{proof}[ Proof of Lemma \ref{lem2}] Given a finite set $M$, we denote by $|M|$ its cardinality. To ease notation, let $i:=i_k$ and $j:=i_{k'}$. We want to prove that either $i^+$ or $j^+$ learns the message at the end of the block $B_b=\{t_b,\dots,t_{b}+2n^C-4\}$. The proof is by contradiction. So, assume neither $i^+$ nor $j^+$ learns the message at the end of the block.

Fix a strategy profile $\sigma' \in \Sigma^*(\sigma)$.  Denote by $S^i$ the stages where player $i$ is deviating from $\sigma$, by $S^{i^-}$ the stages where player $i^{-}$ is deviating, by $S^j$ the stages where player $j$ is deviating and by $S^{j^-}$ the stages where player $j^{-}$ is deviating. From the definition of $\Sigma^*(\sigma)$,  the sets $S^i$, $S^{i^-}$, $S^j$ and $S^{j^-}$ are pairwise disjoints. In particular,
\begin{align}\label{toto1}
|S^i|+|S^{i^-}|+|S^j|+|S^{j^-}| \leq | S^i \cup S^{i^-} \cup S^{j} \cup S^{j^-} | \leq 2n^C-3.
\end{align}

Throughout, for any subset $S$ of $B_b$, we write $\overline{S}$ for its complement in $B_b$. By definition, at all stages in $\overline{S^i}$, player $i$ follows $\sigma_i$ and deviates at all others. Let $\overline{S^i}:=\{\overline{s}_1,\dots,\overline{s}_{\ell}, \dots, \overline{s}_{\ell^*_i}\}$, with $\overline{s}_{\ell}< \overline{s}_{\ell+1}$ for all $\ell$. Notice that $\ell^*_i=(2 n^C-3)-|S^{i}|$. 
\medskip

Assume that $\ell^*_i \geq n^C-1$. Since player $i$ follows the protocol at all stages in $\overline{S^i}$, player $i^+$ observes at least one sequence of messages such that $m$ is broadcasted $n^C-1$ times by player $i$. Consider all sequences $(\overline{s}_{\ell_1},\dots, \overline{s}_{\ell_{n^C-1}})$ of distinct elements of $\overline{S^i}$ such that all sequences have $n^C-1$ consecutive elements, that is, if  $\overline{s}_\ell$ and $\overline{s}_{\ell'}$ are elements of the sequence, so are all $\overline{s}_{\ell{''}}$ satisfying $\overline{s}_{\ell} <\overline{s}_{\ell{''}} < \overline{s}_{\ell'}$. By construction, there are $(\ell^*_i+1)-(n^C-1)=n^C-|S^i|-1$ such sequences. All these sequences have different starting stages and, therefore, different ending stages. Recall that player $i$ broadcasts $m$ at all stages of these sequences.\medskip 

Fix the sequence $(\overline{s}_{\ell_1},\dots, \overline{s}_{\ell_{n^C-1}})$. The protocol specifies that player $i^+$  learns  $m$ if and only if he has not received the correct authentication key $x_{i}^{s_{\ell_1}}$ from player $i^{++}$ by the stage $\overline{s}_{\ell_{n^C-1}}$. Therefore, player $i^+$ does \emph{not} learn $m$ only if  player $i^-$ broadcasts the correct authentication key at some stage $s > \overline{s}_{\ell_1}$; the other players do not know the authentication key and the probability of guessing it correctly is zero. Moreover, since the protocol requires player $i^-$ to broadcast $x_{i}^{s_{\ell_1}}$ only if player $i$ broadcasts $m' \neq m$ at stage $\overline{s}_{\ell_1}$, which he does not, player $i^-$ must be deviating. Therefore, $s \in S^{i^-}$. \medskip

Remember that player $i^-$ can broadcast at most one authentication key about player $i$ at each stage. Therefore, since there are $n^C-|S^i|-1$ such sequences,  player $i^-$ must deviate at least $n^C-|S^i|-1$ times for player $i^+$ to not learn $m$, that is, 
\begin{align}
|S^{i^-}|\geq n^C-|S^i|-1.
\end{align}
It follows that
\begin{align}\label{toto2}
|S^i|+|S^{i^-}| \geq n^C-1.
\end{align}

Assume now that $\ell^*_i < n^C-1$. We have that $|S^i|=(2n^C-3)-\ell^*_i > 2n^C-3-n^C-1=n^C-2$, hence $|S^i|\geq n^C-1$. Inequality (\ref{toto2}) is also satisfied.\\

A symmetric argument applies to the pair of player $j$ and player $j^{-}$, hence 
\begin{align}\label{toto3}
|S^j|+|S^{j^-}| \geq n^C-1,
\end{align}
since player $j^+$ does not learn $m$ either. 
Summing  Equations (\ref{toto2}) (\ref{toto3}), we obtain that
\begin{align}
|S^i|+|S^{i^-}|+|S^j|+|S^{j^-}|  \geq 2n^C-2,
\end{align}
a contradiction with Equation (\ref{toto1}). This completes the proof of Lemma \ref{lem2}. 
\end{proof}

To conclude the proof, it is enough to invoke Lemma \ref{lem1} and \ref{lem2}, which guarantees that the receiver learns almost surely the message broadcasted by the sender.

\section{Proof of Theorem \ref{theo1}}

Let $(\sigma^*,\tau^*)$ be an equilibrium of the direct communication game with equilibrium distribution $\mu$. We use  the protocol $(\boldsymbol{\mathcal{M}},\sigma)$ used to prove Theorem \ref{theo2} to prove Theorem \ref{theo1}. \medskip

The communication game is $\Gamma(\boldsymbol{\mathcal{M}},\mathcal{N})$. We now describe the strategies. If the state is $\omega$, the sender chooses a message $m$ in the support of $\sigma^*(\cdot |\omega) \in \Delta(M)$ and then follows the protocol $\sigma_S$, that is, the sender first broadcasts $m$ to the intermediaries $(i_1,j_1)$, i.e., his two immediate successors on the two disjoint paths to the receiver, and then follows $\sigma_S$ at all stages $t \geq 2$. Note that $\sigma_S$ is independent of $\omega$ at all stages $t \geq 2$. Similarly,  intermediary $i \in I$ follows the strategy $\sigma_i$. Lastly, the receiver follows $\sigma_R$ until he learns the message $m$, at which stage he takes an action in $A$ according to $\tau^*(\cdot |m) \in \Delta(A)$. \medskip 

Clearly, since the receiver learns the message $m$ at all histories consistent with unilateral deviations, no intermediary has an incentive to deviate since it would result in the same expected payoff. Similarly, the sender has no incentives to deviate since the sender selects the message $m$ according to the equilibrium strategy of the direct communication game and, conditional on broadcasting $m$ at the first stage, the receiver receives $m$ at all histories consistent with unilateral deviations, including deviations by the sender. Finally, the receiver has no incentive to deviate either. If he stops the game earlier, then his expected payoff is weakly lower as a consequence of Blackwell's theorem. Indeed, the only informative message about $\omega$ is $m$ and stopping earlier is a garbling of $m$. \medskip

\section{Sequential rationality} \label{seq-rat}
As already argued in the text, sequential rationality is guaranteed at all histories consistent with at most one intermediary deviating at each stage of the communication game. We therefore focus our attention on all other histories, i.e., histories not in $\mathcal{H}(\sigma)$.

We first consider all intermediaries $(i_1,\dots,i_K)$ and $(j_1,\dots,j_{K'})$. We treat the sender and receiver separately. 

\textbf{Rebooting strategies.} We say that player $i$ reboots his strategy at period $t$ if, from any history $h_i^t \notin \mathcal{H}_{i}(\sigma)$ onwards, he follows the protocol as if he knows that the message is $m_0$. That is, at history $h_i^t$, he broadcasts the message $m_0$, an authentication key $x_i^t$, and random triplets $(j,t_j,x_j^{t_j})$, $j\in \mathcal{N}_i$. At all subsequent histories consistent with at most one intermediary deviating from the protocol at each stage, player $i$ continues to follow the protocol. That is, player $i$ continues to broadcast $m_0$,  authentication keys and triplet $(j,t_j,x_j^{t_j})$, as specified by the protocol when a player knows a message (here, it is $m_0$). At all other histories, player $i$ reboots yet again his strategy, that is, player $i$ continues to broadcast $m_0$, authentication keys and triplets, as if the multilateral deviation hadn't taken place.

\textbf{Beliefs.} At history $h_i^t \notin \mathcal{H}_{i}(\sigma)$, player $i$ believes that all other players on the same side of the circle reboot their strategies, while players on the other side of the circle continue to follow the protocol. (Here, the sender and receiver are assumed to be on the other side of the circle.) In other words, player $i$ believes that all other players on the same side of the circle have also observed a multilateral deviation, while players on the other side have observed no deviations.

We now consider the sender. At all histories $h_S^t$, the sender continues to follow the protocol as if the observed multilateral deviations hadn't happened. However, he believes that all intermediaries reboot their strategies at period $t$, while the receiver continues to follow the protocol.

 Finally, we consider the receiver. The receiver continues to validate messages as he does in the protocol, i.e., he tests sequences of messages of length $n^C-1$ received by his two predecessors and  validates a message if he has received a sequence of $n^C-1$ identical copies of the message and has not received the correct authentication on time (see the construction of the protocol for details). To complete the construction of the strategies, we assume that if the receiver validates $m \in M$ and $m_0 \notin M$, then he plays $\tau^*(m)$. Similarly, if he validates two different messages $(m,m') \in M \times M$ or $(m_0,m_0)$ or no messages at all, he plays a best reply to his prior.  At all histories, the receiver continues  to follow the protocol as if the observed deviations hadn't happened. He believes that all intermediaries reboot their strategies, while the sender continues to follow the protocol. 
 
 \textbf{Sequential rationality.} At history $h_i^{t}\notin \mathcal{H}_{i}(\sigma)$, an intermediary expects the receiver to validate a message $m \in M$ from the other side and to validate the message $m_0$ from his side. Since the receiver takes the decision $\tau^*(m)$ when validating the messages $m \in M$ and $m_0 \notin M$, the intermediary cannot deviate profitably (as, regardless  of his play, the receiver validates $m$ from the other side). Therefore, rebooting the strategy is optimal. Similarly, since the sender expects the intermediaries $i_{K}$ and $j_{K'}$ to reboot their strategies, he expects the receiver to play $a^*$ and, therefore, cannot profitably deviate. The same applies to the receiver.

\section{Proof of Theorem \ref{theo3}}\label{prooftheo3}
The proof is constructive and relies extensively on Theorem \ref{theo2}. The main idea is to generate a jointly controlled lottery between the sender and one of the two intermediaries $i_1$ or $j_1$ to generate a recommendation, which is then reliably transmitted to the receiver. \medskip 

We start with a formal definition of jointly controlled lotteries.

\begin{definition}
Let $\mu\in \Delta(A)$. A jointly controlled lottery generating $\mu$ is a triple $(X,Y,\phi)$ such that
\begin{itemize}
\item[-] $X$ is a measurable function from a probability space $(U,\mathcal{U},\mathbb{P})$ to $([0,1],\mathbb{B}_{[0,1]})$,
\item[-] $Y$ is a measurable function from a probability space $(U,\mathcal{U},\mathbb{P})$ to $([0,1],\mathbb{B}_{[0,1]})$,
\item[-] $X$ and $Y$ are independent,
\item[-] and $\phi$ is a measurable function from $([0,1]^2,\mathbb{B}_{[0,1]^2})$ to $(A,\mathbb{B}_{A})$,
\end{itemize}
such that 
\begin{itemize}
\item[(i)] the law of $Z:=\phi(X,Y)$ is $\mu$,
\item[(ii)] for every $x\in [0,1]$, the law of $\phi(x,Y)$ is $\mu$,
\item[(iii)] and for every $y\in [0,1]$, the law of $\phi(X,y)$ is $\mu$.
\end{itemize}
\end{definition}

As explained in the main text, when $A$ is finite, any distribution $\mu$ over $A$ can be generated by a jointly controlled lottery.   The idea is to partition the interval $[0,1]$ into $|A|$ sub-intervals, where the length of the sub-interval associated with $a$ is $\mu(a)$. Let $f: [0,1] \rightarrow A$, where $f(r)=a$ if $r$ is in the sub-interval associated with $a$. Note that $f^{-1}(a)$ is a Borel set and has measure $\mu(a)$. Consider then 
 two uniform random variables $X$ and $Y$. The key observation  to make is that the sums $X+Y \mod{[0,1]}$, $x+Y \mod{[0,1]}$, and $X+y \mod{[0,1]}$ are also uniformly distributed on $[0,1]$, regardless of the values of $x$ and $y$. Therefore, if we let $\phi(x,y) =a$ if $x+y \mod{[0,1]} \in f^{-1}(a)$, then the triplet $(X,Y,\phi)$ generates $\mu$. The next proposition states that this construction generalizes to arbitrary complete and separable metric space $A$.

\begin{proposition}\label{prop:lottery}
For any $\mu\in \Delta(A)$, there exists a jointly controlled lottery $(X,Y,\phi)$ that generates $\mu$.
\end{proposition}

\begin{proof}[Proof of Proposition \ref{prop:lottery}]
Let $\lambda$ be the Lebesgue measure on $[0,1]$. From the fundamental principle of simulation (Theorem A.3.1, p. 38 in Bouleau and Lepingle, 1993), there exists a Borel function $f:[0,1] \rightarrow A$ such that $f$ is $\lambda$-a.e. continuous and
\[
\forall E \in \mathbb{B}_{A},\ \lambda \circ f^{-1}(E):=\lambda(f^{-1}(E))=\mu(E).\]
To complete the proof, let  $X$ and $Y$ be two random variables with uniform distribution on $[0,1]$ and define $\phi$ as
\[\phi(x,y)=f(x+y\mod_{[0,1]}),\]
for all $(x,y) \in [0,1]$. It is routine to verify that the triplet $(X,Y,\phi)$ is a jointly controlled lottery, which generates $\mu$.
\end{proof}
\medskip

We now explain how to robustly implement mediated communication on the network $\mathcal{N}$. 
Let $\tau^*: \Omega \rightarrow \Delta(A)$ be a canonical communication equilibrium of the direct communication game. From Proposition \ref{prop:lottery}, for each $\omega$, there exists a jointly controlled lottery $(X_\omega,Y_\omega,\phi_\omega)$, which generates $\tau^*(\omega)$.

As in the proof of Theorem \ref{theo2}, we let $\mathcal{P}$ be the players on the two disjoint path from the sender to the receiver. The communication game is as follows: 

\begin{description}
\item[$t=1$] The sender truthfully broadcasts the state $\omega$ to the intermediaries $i_1$ and $j_1$. 
\item[$t=2$] The sender and intermediary $i_1$ draw a random number in [0,1] each. Let $x$ (resp., $y$) the number drawn by the sender (resp., intermediary $i_1$). The sender and the intermediary $i_1$ broadcast their random number. 
\item[$t=3,\dots,2+ (2n^C-3)(n^C-3)$] The players $p \in P$ execute the protocol with $i_1$ in the role of the sender and $j_1$ in the role of the receiver and the message to be transmitted is $y$. 
\item[$t= 3+ (2n^C-3)(n^C-3)$]  The sender and the intermediaries $i_1$ and $j_1$ outputs a recommendation $a \in A$ according to $(X_\omega,Y_\omega,\phi_\omega)$, that is, the recommendation is $\phi_{\omega}(x,y)$ with $\omega$ the state broadcasted at $t=1$. The three of them truthfully broadcast the recommendation. 
\item[$t=4+ (2n^C-3)(n^C-3), \dots, 2+ 2(2n^C-3)(n^C-3)$] The players execute in parallel and independently three copies of the protocol with $S$, $i_1$ and $j_1$ in the role of the sender, respectively, and the message to be transmitted is the recommendation $a$. At the last stage, the receiver follows the recommendation made a majority of times, if any. (If there is no majority, then he chooses an arbitrary action.)
\end{description}

Since at stage $t= 3+ (2n^C-3)(n^C-3)$, at most one of the three ``senders'' can deviate, the correct recommendation is sent at least twice. It follows that if the receiver is obedient,  the receiver chooses the correct action at all histories consistent with unilateral deviations. Moreover, since the receiver observes neither $\omega$ nor $x$,  he has no additional information about the state than in the direct communication game, hence he has an incentive to be obedient.

\end{document}